\documentclass[prl,showpacs,amsmath,amssymb,twocolumn]{revtex4}
\usepackage{mathrsfs}
\usepackage{amsfonts}

\usepackage{amsthm}
\usepackage{dcolumn}
\usepackage{bm}
\usepackage{graphicx}

\newtheorem{lemma}{Lemma}
\newtheorem{theorem}{Theorem}
\newtheorem{corollary}{Corollary}

\newcommand\ket[1]{\ensuremath{|#1\rangle}}
\newcommand\bra[1]{\ensuremath{\langle#1|}}

\newcommand\oprod[2]{\ensuremath{|#1\rangle\langle#2|}}
\newcommand\tr{\mathop{\rm tr}\nolimits}

\begin{document}

\title{Separability Criterion for Multipartite Pure States}

\author{Zongwen Yu}
  \email{yzw04@mails.tsinghua.edu.cn}
\author{Su Hu}
  \email{hus04@mails.tsinghua.edu.cn}
\affiliation{ Department of Mathematical Sciences, Tsinghua
University, Beijing 100084, China }


\begin{abstract}
In this letter, we give out some effective criterions which can be
used to judge the separability of multipartite pure states. We
obtain the relationship between separability and Schmidt
decomposable of multipartite pure states in
Theorem~\ref{thm:thmSchmidt}. The first criterion derived from
Theorem~\ref{thm:thmDet} dose not need the Schmidt decomposition
which is hard to find for multipartite states.
Theorem~\ref{thm:thmRank1} is more profound which can be used to
deduce Corollary~\ref{cor:corSubmatrix} which is one of the main
results in~\cite{DafaLi2006}. Finally, we give out an algorithm
which can be used to judge the separability of multipartite pure
states effectively.
\end{abstract}


\maketitle

The state is one of the fundamental concepts in quantum computation
and information which can be classified into pure state and mixed
state. In view of the purification~\cite{NC2000}, we always pay
attention to the pure states with related technique for quantum
computation and quantum information. The pure states can be
classified into pure separable states and pure entanglement states.
Entanglement is a valuable physical resource for accomplishing many
useful quantum computing and quantum information processing
tasks~\cite{NC2000}. For certain tasks such as superdense
coding~\cite{Ben1992} and quantum teleportation~\cite{Ben1993}, it
has been demonstrated that entanglement is an indispensable
ingredient. For many other tasks entanglement is also used to
enhance the efficiency~\cite{Childs2000,Acin2001,DA2001,ZFJi2006}.
The question of quantifying entanglement of multipartite quantum
states is fundamental to the whole field of quantum information and
in general to the physics of multicomponent quantum systems.
Separability is a theoretical foot stone to define the measures of
entanglement. As a result, the problem of separability, that is
whether a quantum state is separable or entangled, is fundamental.
Some authors have given some separability criterions, valid under
certain
conditions~\cite{Peres1996,Wu2000,Vcoffman2001,JEisert2004,YuChSh2005,DafaLi2006}.
In this letter, we will devote to the separability criterions for
any multipartite pure state.

Let $\ket{\psi}$ be a pure state of a composite system AB possessed
by Alice and Bob, then we know that $\ket{\psi}$ has Schmidt
decomposition
$\ket{\psi}=\sum_{i}{\lambda_{i}\ket{i^{A}}\ket{i^{B}}}$~\cite{NC2000}.
It is also known that a state of a bipartite system is separable if
and only if it has Schmidt number 1~\cite{NC2000}. So we can judge a
bipartite pure state to be separable or not by computing its Schmidt
number. Unfortunately, the Schmidt decomposition does not always
exist for any multi($n$)-partite pure state when $n>2$. In order to
extend the forward criterion to any multipartite pure state, some
pioneers have paid their attentions to the conditions of the
occurrence of Schmidt decomposition for multipartite pure states.
Peres~\cite{Peres1995} presented a necessary and sufficient
condition for the occurrence of Schmidt decomposition for a
tripartite pure state and~\cite{Peres1996} showed that the
positivity of the partial transpose of a density matrix is a
necessary condition for separability. Unfortunately, this criterion
is only necessary for a pure state to be separable, but not
sufficient. Thapliyal~\cite{Thap1999} showed that a multipartite
pure state is Schmidt decomposable if and only if the deduced
density matrices obtained by tracing out any party are separable. We
should note that the separable states in Thapliyal's criterion
contain the pure separable states and the mixed separable states.
Making use of the Schmidt decomposition we will get the first
criterion.

Let $\ket{\psi}$ be a pure state in a $d$ dimensional $n$-partite
quantum system $H$, which is composed by the subsystems
$H^{A_{1}},H^{A_{2}},\cdots,H^{A_{n}}$. So we have
$H=H^{A_{1}}\otimes H^{A_{2}}\otimes\cdots\otimes
H^{A_{n}},d=\prod_{k=1}^{n}{d_{k}}$ and denote
$D=\max{\{d_{1},d_{2},\cdots,d_{n}\}}$ where $d_{k}$ is the
dimension of $k$-th subsystem $H^{A_{k}}$ for $k=1,2,\cdots,n$.
Because $\ket{\psi}$ is a pure state in $H$, we have
$\ket{\psi}=\sum\limits_{i_{1}=1}^{d_{1}}{\sum\limits_{i_{2}=1}^{d_{2}}
{\cdots\sum\limits_{i_{n}=1}^{d_{n}}{a_{i_{1}i_{2}\cdots
i_{n}}\ket{e_{i_{1}}^{A_{1}}}\ket{e_{i_{2}}^{A_{2}}}\cdots\ket{e_{i_{n}}^{A_{n}}}}}}$,
where $\{\ket{e_{i_{k}}^{A_{k}}}\}_{i_{k}=1}^{d_{k}}$ are the
orthonormal basis of subsystems $H^{A_{k}}(k=1,2,\cdots,n)$ and
$a_{i_{1}i_{2}\cdots i_{n}}\in\mathbb{C}$ are the amplitudes with
$\sum\limits_{i_{1}=1}^{d_{1}}{\sum\limits_{i_{2}=1}^{d_{2}}{\cdots\sum\limits_{i_{n}=1}^{d_{n}}{|a_{i_{1}i_{2}\cdots
i_{n}}|^{2}}}}=1$. By definition, a $n$-partite pure state
$\ket{\psi}$ is separable if and only if there exits $n$ pure states
$\ket{\psi^{A_{k}}}(k=1,2,\cdots,n)$ respectively belonging to
subsystems $H^{A_{k}}(k=1,2,\cdots,n)$ and
$\ket{\psi}=\ket{\psi^{A_{1}}}\otimes\ket{\psi^{A_{2}}}\otimes\cdots\otimes\ket{\psi^{A_{n}}}$.
Generalizing the separability condition for a bipartite pure state
we have the following theorem.

\begin{theorem}\label{thm:thmSchmidt}
A $n$-partite pure state $\ket{\psi}$ is separable if and only if it
is Schmidt decomposable and has Schmidt number 1.
\end{theorem}

\begin{proof}
If a $n$-partite pure state $\ket{\psi}$ is separable, then we have
\begin{equation}\label{eq:eqSep}
\ket{\psi}=\ket{\psi^{A_{1}}}\otimes\ket{\psi^{A_{2}}}\otimes\cdots\otimes\ket{\psi^{A_{n}}},
\end{equation}
where $\ket{\psi^{A_{k}}}$ is a pure state in $H^{A_{k}}$ for
$k=1,2,\cdots,n$ respectively. Let
$\ket{e_{1}^{A_{k}}}=\ket{\psi^{A_{k}}}(k=1,2,\cdots,n)$ and choose
$\{\ket{e_{i_{k}}^{A_{k}}}\}_{i_{k}=2}^{d_{k}}$ to make
$\{\ket{e_{i_{k}}^{A_{k}}}\}_{i_{k}=1}^{d_{k}}$ to be the
orthonormal basis of subsystems $H^{A_{k}}$ for $k=1,2,\cdots,n$.
Taking into account Eq.~\eqref{eq:eqSep}, if we choose
$\lambda_{1}=1,\lambda_{2}=\lambda_{3}=\cdots=\lambda_{D}=0$ and
$\ket{e_{i_{k}}^{A_{k}}}=0(i_{k}=d_{k}+1,d_{k}+2,\cdots,D)$, then we
have
\begin{equation}\label{eq:eqSchmidt}
\ket{\psi}=\sum_{s=1}^{D}{\lambda_{s}\ket{e_{s}^{A_{1}}}\ket{e_{s}^{A_{2}}}\cdots\ket{e_{s}^{A_{n}}}},
\end{equation}
which is actual the Schmidt decomposition and we know that the
Schmidt number of $\ket{\psi}$ is 1.

On the other hand, if a $n$-partite pure state $\ket{\psi}$ is
Schmidt decomposable and has Schmidt number 1. Supposing that
$\lambda_{1}\neq 0$, then we have
$\ket{\psi}=\sum\limits_{s}{\lambda_{s}\ket{e_{s}^{A_{1}}}\ket{e_{s}^{A_{2}}}\cdots\ket{e_{s}^{A_{n}}}}=
\ket{e_{1}^{A_{1}}}\ket{e_{1}^{A_{2}}}\cdots\ket{e_{1}^{A_{n}}}$,
where $\ket{e_{1}^{A_{k}}}$ is a pure state in $H^{A_{k}}$ for
$k=1,2,\cdots,n$.
\end{proof}

Taking into account Theorem~\ref{thm:thmSchmidt}, we can judge the
separability of any $n$-partite pure state by finding its Schmidt
decomposition. In order to obtain the Schmidt decomposition of a
$n$-partite pure state $\ket{\psi}$, we need to compute (1) the
density operator $\rho_{\ket{\psi}}$, (2) the reduced density
operators $\rho_{\ket{\psi}}^{A_{k}}(k=1,2,\cdots,n)$ and (3) the
eigenvalues $\lambda_{i_{k}}$ of
$\rho_{\ket{\psi}}^{A_{k}}(k=1,2,\cdots,n)$. However it is hard to
compute all the eigenvalues of high dimensional density operators
exactly. Theorem~\ref{thm:thmSchmidt} does not give us an effective
criterion. In the following, we will introduce some effective
separability criterions.

Any quantum state can be represent by a density operator. For a
$n$-partite pure state $\ket{\psi}$ in $H$, its density operator is
$\rho_{\ket{\psi}}={\oprod{\psi}{\psi}}$. Taking the partial trace
operation on each $k$-th subsystem $H^{A_{k}}$, we get reduced
density operators
$\rho_{\ket{\psi}}^{A_{k}}=\tr_{H^{A_{k}}}{\left(\rho_{\ket{\psi}}\right)}$
(with $O\left(\prod\limits_{{s=1}\atop{s\neq
k}}^{n}{d_{s}^{2}d_{k}}\right)$ multiplication operations and
$O\left(\sum_{t=1}^{n}{d_{t}}\right)$ plus operations in the worst
case) for $k=1,2,\cdots,n$.

In order to judge the separability of a $n$-partite pure state there
is no need to find the Schmidt decomposition since we have the
following theorem.

\begin{theorem}\label{thm:thmDet}
A $n$-partite pure state $\ket{\psi}$ is separable if and only if
\begin{equation}\label{eq:eqDet}
\det{\left(M_{\ket{\psi}}^{A_{k}}-E_{k}\right)}=0,\quad
k=1,2,\cdots,n.
\end{equation}
where $M_{\ket{\psi}}^{A_{k}}$ are the matrices related to the
reduced density operators $\rho_{\ket{\psi}}^{A_{k}}$ for
$k=1,2,\cdots,n$. $E_{k}$ are the identity matrices and have the
same dimension with $M_{\ket{\psi}}^{A_{k}}$ for $k=1,2,\cdots,n$.
\end{theorem}

By using the characterizations of density operators (trace and
positivity conditions)~\cite{NC2000}, we can easily prove the
following lemma.

\begin{lemma}\label{lem:lemRank1}
The rank of matrices $M_{\ket{\psi}}^{A_{k}}(k=1,2,\cdots,n)$ are
all equal to 1 if and only if Eq.~\eqref{eq:eqDet} can be obtained.
\end{lemma}

Then we can give out the proof of Theorem~\ref{thm:thmDet} by using
Lemma~\ref{lem:lemRank1}.

\begin{proof}
Taking into account Theorem~\ref{thm:thmSchmidt} and
Lemma~\ref{lem:lemRank1}, we need only to prove that the $n$-partite
pure state is Schmidt decomposable and has Schmidt number 1 if and
only if the rank of matrices $M_{\ket{\psi}}^{A_{k}}$ are equal to 1
for $k=1,2,\cdots,n$.

If the $n$-partite pure state $\ket{\psi}$ is Schmidt decomposable
and has Schmidt number 1, then we have
$\ket{\psi}=\ket{e_{1}^{A_{1}}}\ket{e_{1}^{A_{2}}}\cdots\ket{e_{1}^{A_{n}}}$.
We can calculate the density operator $\rho_{\ket{\psi}}
={\oprod{\psi}{\psi}}
=\ket{e_{1}^{A_{1}}}\ket{e_{1}^{A_{2}}}\cdots\ket{e_{1}^{A_{n}}}
\bra{e_{1}^{A_{1}}}\bra{e_{1}^{A_{2}}}\cdots\bra{e_{1}^{A_{n}}}$ and
reduced density operators
\begin{widetext}
$\rho_{\ket{\psi}}^{A_{k}}
=\tr_{H^{A_{k}}}{\left(\rho_{\ket{\psi}}\right)}
=\ket{e_{1}^{A_{1}}}\ket{e_{1}^{A_{2}}}\cdots
\ket{e_{1}^{A_{k-1}}}\ket{e_{1}^{A_{k+1}}}\cdots\ket{e_{1}^{A_{n}}}
\bra{e_{1}^{A_{1}}}\bra{e_{1}^{A_{2}}}\cdots
\bra{e_{1}^{A_{k-1}}}\bra{e_{1}^{A_{k+1}}}\cdots\bra{e_{1}^{A_{n}}},
(k=1,2,\cdots,n)$.
\end{widetext}
So we have
\begin{equation*}
M_{\ket{\psi}}^{A_{k}}=\left(
\begin{array}{cccc}
  1  &  0  &  \cdots  &  0  \\
  0  &  0  &  \cdots  &  0  \\
  \vdots  &  \vdots  &  \ddots  &  \vdots  \\
  0  &  0  &  \cdots  &  0
\end{array}\right),
\quad k=1,2,\cdots,n.
\end{equation*}
This means that $M_{\ket{\psi}}^{A_{k}}(k=1,2,\cdots,n)$ only have 1
to be their nonzero eigenvalues. So Eq.~\eqref{eq:eqDet} is
obtained.

On the other hand, all the matrices
$M_{\ket{\psi}}^{A_{k}}(k=1,2,\cdots,n)$ only have 1 to be their
nonzero eigenvalues. Denote $\ket{f_{1}^{A_{k}}}$ to be the
eigenvectors of $M_{\ket{\psi}}^{A_{k}}$ corresponding to the
eigenvalues 1 for $k=1,2,\cdots,n$ respectively. Then we have
$\ket{\psi}=\ket{f_{1}^{A_{1}}}\ket{f_{1}^{A_{2}}}\cdots\ket{f_{1}^{A_{n}}}$
which is the Schmidt decomposition of $\ket{\psi}$ with Schmidt
number 1. So $\ket{\psi}$ is separable.
\end{proof}

Theorem~\ref{thm:thmDet} can be used to judge the separability of
any n-partite pure state. For example, $n$-cat state (in honor of
Schr$\ddot{o}$dinger's cat ) is
$\ket{\psi}={\frac{\sqrt{2}}{2}}\left(\ket{0^{\otimes
n}}+\ket{1^{\otimes n}}\right)$ with the
Einstein-Podolsky-Rosen-Bohm pair
${\frac{\sqrt{2}}{2}}\left(\ket{00}+\ket{11}\right)$ when $n=2$ and
the Greenberger-Hone-Zeilinger-Mermin state
${\frac{\sqrt{2}}{2}}\left(\ket{000}+\ket{111}\right)$ when $n=3$.
We have the density operator
$\rho_{\ket{\psi}}={\frac{1}{2}}\left({\oprod{0^{\otimes
n}}{0^{\otimes n}}}+\ket{0^{\otimes n}}\bra{1^{\otimes
n}}+\ket{1^{\otimes n}}\bra{0^{\otimes n}}+{\oprod{1^{\otimes
n}}{1^{\otimes n}}}\right)$ and reduced operators
$\rho_{\ket{\psi}}^{A_{k}}={\frac{1}{2}}\left({\oprod{0^{\otimes
(n-1)}}{0^{\otimes (n-1)}}}+{\oprod{1^{\otimes (n-1)}}{1^{\otimes
(n-1)}}}\right)$ which means that
\begin{equation*}
  M_{\ket{\phi}}^{A_{k}}=\left(
  \begin{array}{ccccc}
    {\frac{1}{2}}  &  0  &  \cdots  &  0  &  0  \\
    0  &  0  &  \cdots  &  0  &  0  \\
    \vdots  &  \vdots  &  \ddots  &  \vdots  &  \vdots  \\
    0  &  0  &  \cdots  &  0  &  0  \\
    0  &  0  &  \cdots  &  0  &  {\frac{1}{2}}
  \end{array}\right),
  \quad k=1,2,\cdots,n.
\end{equation*}

Then we have
$\det{\left(M_{\ket{\psi}}^{A_{k}}-E_{k}\right)}={\frac{1}{4}}\neq
0(k=1,2,\cdots,n)$ which means that the $n$-cat state is an
entanglement pure state by Theorem~\ref{thm:thmDet}.

Theorem~\ref{thm:thmDet} give us an separability criterion for any
$n$-partite pure state. The total number of times the criterion has
to run, in the worst case, is $O(nd^{3})$ with the most operations
being used for computing the determinants.

For a $n$-partite pure state, we have
$\ket{\psi}=\sum\limits_{i_{1}=1}^{d_{1}}{\sum\limits_{i_{2}=1}^{d_{2}}
{\cdots\sum\limits_{i_{n}=1}^{d_{n}}{a_{i_{1}i_{2}\cdots
i_{n}}\ket{e_{i_{1}}^{A_{1}}}\ket{e_{i_{2}}^{A_{2}}}\cdots\ket{e_{i_{n}}^{A_{n}}}}}}$.
Let $M_{k}$ are $\left(\prod\limits_{{s=1}\atop{s\neq
k}}^{n}{d_{s}}\right)\times d_{k}$ matrices of the amplitudes
$a_{i_{1}i_{2}\cdots i_{n}}(i_{s}=1,2,\cdots,d_{s})$ of the form
\begin{widetext}
 \begin{equation}\label{eq:eqMk}
  M_{k}=\left(
  \begin{array}{cccc}
    a_{11\cdots 111\cdots 111}  &  a_{11\cdots 121\cdots 111}  &
    \cdots  &  a_{11\cdots 1d_{k}1\cdots 111}  \\
    a_{11\cdots 111\cdots 112}  &  a_{11\cdots 121\cdots 112}  &
    \cdots  &  a_{11\cdots 1d_{k}1\cdots 112}  \\
    \vdots  &  \vdots  &  \ddots  &  \vdots  \\
    a_{11\cdots  111\cdots 11d_{n}}  &  a_{11\cdots 121\cdots
    11d_{n}}  &  \cdots  &  a_{11\cdots 1d_{k}1\cdots 11d_{n}}  \\
    \vdots  &  \vdots  &  \ddots  &  \vdots  \\
    a_{i_{1}i_{2}\cdots i_{k-1}1i_{k+1}\cdots i_{n}}  &
    a_{i_{1}i_{2}\cdots i_{k-1}2i_{k+1}\cdots i_{n}}  &  \cdots  &
    a_{i_{1}i_{2}\cdots i_{k-1}d_{k}i_{k+1}\cdots i_{n}}  \\
    \vdots  &  \vdots  &  \ddots  &  \vdots  \\
    a_{d_{1}d_{2}\cdots d_{k-1}1d_{k+1}\cdots d_{n}}  &
    a_{d_{1}d_{2}\cdots d_{k-1}2d_{k+1}\cdots d_{n}}  &  \cdots  &
    a_{d_{1}d_{2}\cdots d_{k-1}d_{k}d_{k+1}\cdots d_{n}}
  \end{array}\right), \quad k=1,2,\cdots,n.
 \end{equation}
\end{widetext}

We have the following lemma without proof.

\begin{lemma}\label{lem:lemDenmatrix}
The reduced density matrices $M_{\ket{\psi}}^{A_{k}}$ satisfy the
following equations.
\begin{equation}\label{eq:eqDenmatrix}
  M_{\ket{\psi}}^{A_{k}}=M_{k}M_{k}^{\dagger}, \quad k=1,2,\cdots,n.
\end{equation}
where $M_{k}(k=1,2,\cdots,n)$ are given by Eq.~\eqref{eq:eqMk} and
$M_{k}^{\dagger}$ are the Hermitian of $M_{k}(k=1,2,\cdots,n)$.
\end{lemma}

To avoid the Schmidt decomposition, by using
Lemma~\ref{lem:lemRank1} and Lemma~\ref{lem:lemDenmatrix} we also
obtain the following theorem.

\begin{theorem}\label{thm:thmRank1}
A $n$-partite pure state $\ket{\psi}$ is separable if and only if
the rank of matrices $M_{k}$ are equal to 1 for $k=1,2,\cdots,n$.
\end{theorem}

\begin{proof}
Theorem~\ref{thm:thmDet} and Lemma~\ref{lem:lemRank1} tell us that a
$n$-partite pure state $\ket{\psi}$ is separable if and only if the
rank of matrices $M_{\ket{\psi}}^{A_{k}}(k=1,2,\cdots,n)$ are equal
to 1. Taking into account Lemma~\ref{lem:lemDenmatrix}, we know that
the rank of matrices $M_{\ket{\psi}}^{A_{k}}(k=1,2,\cdots,n)$ are
equal to 1 if and only if the rank of $M_{k}(k=1,2,\cdots,n)$ are
equal to 1.
\end{proof}

\begin{corollary}\label{cor:corSubmatrix}
A $n$-partite pure state $\ket{\psi}$ is separable if and only if
the determinants of all the $2\times 2$ submatrices of
$M_{k}(k=1,2,\cdots,n)$ are zeros.
\end{corollary}

Using Theorem~\ref{thm:thmRank1}, we can easily prove
Corollary~\ref{cor:corSubmatrix} which is the main result of Dafa Li
in~\cite{DafaLi2006}. In fact, Theorem~\ref{thm:thmRank1} and
Corollary~\ref{cor:corSubmatrix} are equivalent, but this does not
mean they have the same efficiency. Corollary~\ref{cor:corSubmatrix}
can give us another criterion with $O(nd^{2})$ times to be used in
the worst case. By using Theorem~\ref{thm:thmRank1}, we obtain a
more effective criterion.

\begin{corollary}\label{cor:corLineardep}
A $n$-partite pure state $\ket{\psi}$ is separable if and only if
\begin{eqnarray}\label{eq:eqLineardep}
  & & a_{i_{1}i_{2}\cdots i_{k-1}si_{k+1}\cdots i_{n}}\cdot
  M_{k}^{t}=a_{i_{1}i_{2}\cdots i_{k-1}ti_{k+1}\cdots i_{n}}\cdot
  M_{k}^{s}, \nonumber \\
  & & (t=1,2,\cdots,s-1,s+1,\cdots,d_{k};k=1,2,\cdots,n).
\end{eqnarray}
where $a_{i_{1}i_{2}\cdots i_{k-1}si_{k+1}\cdots i_{n}}$ are any
fixed nonzero elements in matrices $M_{k}$ for $k=1,2,\cdots,n$
respectively and $M_{k}^{t}$ is the $t$-th $(t=1,2,\cdots,d_{k})$
column of the matrix $M_{k}$ for $k=1,2,\cdots,n$.
\end{corollary}

\begin{proof}
Taking into account Theorem~\ref{thm:thmRank1}, we only need to
prove that the rank of matrices $M_{k}(k=1,2,\cdots,n)$ are equal to
1 if and only if Eq.~\eqref{eq:eqLineardep} is true.

If the rank of matrices $M_{k}(k=1,2,\cdots,n)$ are equal to 1, then
any two columns $M_{k}^{s},M_{k}^{t}(s,t=1,2,\cdots,d_{k})$ of
$M_{k}(k=1,2,\cdots,n)$ are linearly dependent. This will imply
Eq.~\eqref{eq:eqLineardep}.

On the other hand, Eq.~\eqref{eq:eqLineardep} implies that any two
columns of $M_{k}(k=1,2,\cdots,n)$ are linearly dependent. Since
$M_{k}(k=1,2,\cdots,n)$ are not zero matrices, We have the rank of
matrices $M_{k}(k=1,2,\cdots,n)$ are equal to 1.
\end{proof}

For example, considering the $n$-cat state
$\ket{\psi}={\frac{\sqrt{2}}{2}}\left(\ket{0^{\otimes
n}}+\ket{1^{\otimes n}}\right)$, we have
\begin{equation*}
  M_{k}=\left(
  \begin{array}{cc}
    {\frac{\sqrt{2}}{2}}  &  0  \\
    0  &  0  \\
    \vdots  &  \vdots  \\
    0  &  0  \\
    0  &  {\frac{\sqrt{2}}{2}}
  \end{array}\right), \quad k=1,2,\cdots,n.
\end{equation*}
We can see ${\frac{\sqrt{2}}{2}}\cdot M_{k}^{2}\neq 0\cdot
M_{k}^{1}$, where $M_{k}^{1},M_{k}^{2}$ are the first and second
columns of $M_{k}(k=1,2,\cdots,n)$. So $n$-cat state is an
entanglement state as we know.

Using Corollary~\ref{cor:corLineardep}, we have an effective
algorithm which can be used to judge the separability of a
$n$-partite pure state. The algorithm contains the following 3
steps. Step1: Construct matrices $M_{k}(k=1,2,\cdots,n)$ from the
$n$-partite pure state $\ket{\psi}$ by Eq.~\ref{eq:eqMk}; Step2:
Deleting the zero columns and the zero rows of the matrices
$M_{k}(k=1,2,\cdots,n)$. The result matrices are also denoted by
$M_{k}(k=1,2,\cdots,n)$. Step 3: If Eq.~\eqref{eq:eqLineardep} does
not true for some
$t(t=1,2,\cdots,s-1,s+1,\cdots,d_{k};k=1,2,\cdots,n)$, then stop and
we get $\ket{\psi}$ is an entanglement pure state, otherwise,
$\ket{\psi}$ is a separable pure state. The total number of times
the algorithm has to run, in the worst case, is $O(nd)$.

In this letter, we have given some necessary and sufficient
conditions for the separability of a multipartite pure state. Using
these conditions we get some criterions which can be used to judge
the separability of a multipartite pure state. Finally, an effective
algorithm deduced from Corollary~\ref{cor:corLineardep} has been
provided.


\end{document}